\documentclass{sig-alternate-05-2015}

\begin{document}

\CopyrightYear{2016} 

\setcopyright{acmlicensed}

\conferenceinfo{HSCC'16,}{April 12 - 14, 2016, Vienna, Austria}

\isbn{978-1-4503-3955-1/16/04}\acmPrice{\$15.00}

\doi{http://dx.doi.org/10.1145/2883817.2883839}

\title{Temporal Logic as Filtering}
%
% You need the command \numberofauthors to handle the 'placement
% and alignment' of the authors beneath the title.
%
% For aesthetic reasons, we recommend 'three authors at a time'
% i.e. three 'name/affiliation blocks' be placed beneath the title.
%
% NOTE: You are NOT restricted in how many 'rows' of
% "name/affiliations" may appear. We just ask that you restrict
% the number of 'columns' to three.
%
% Because of the available 'opening page real-estate'
% we ask you to refrain from putting more than six authors
% (two rows with three columns) beneath the article title.
% More than six makes the first-page appear very cluttered indeed.
%
% Use the \alignauthor commands to handle the names
% and affiliations for an 'aesthetic maximum' of six authors.
% Add names, affiliations, addresses for
% the seventh etc. author(s) as the argument for the
% \additionalauthors command.
% These 'additional authors' will be output/set for you
% without further effort on your part as the last section in
% the body of your article BEFORE References or any Appendices.

\numberofauthors{4} 
\author{
	\alignauthor Alena Rodionova\\
	\affaddr{TU Wien}\\
	\affaddr{Treitlstrasse 3}\\
	\affaddr{Vienna, Austria}\\ 
	\email{alena.rodionova@tuwien.ac.at}
	\alignauthor Ezio Bartocci\\
	\affaddr{TU Wien}\\
	\affaddr{Treitlstrasse 3}\\
	\affaddr{Vienna,\\ Austria}\\
	\email{ezio@cps.tuwien.ac.at}
	\alignauthor Dejan Nickovic\\
	\affaddr{Austrian Institute of Technology}\\
	\affaddr{Donau-City-Strasse 1}\\
	\affaddr{Vienna, Austria}\\
	\email{dejan.nickovic@ait.ac.at}
	\and
	\alignauthor Radu Grosu\\
	\affaddr{TU Wien}\\
	\affaddr{Treitlstrasse 3}\\
	\affaddr{Vienna, Austria}\\
	\email{radu.grosu@tuwien.ac.at}%}
}

\maketitle
\begin{abstract}
	We show that metric temporal logic (MTL)
	%, the extension of linear temporal logic to real time,
	can be viewed as linear time-invariant
	filtering, by interpreting addition, multiplication, and their
	neutral elements, over the idempotent dioid ($\max$,$\min$,0,1). 
	Moreover,
	by interpreting these operators over the field of reals
	(+,$\times$,0,1), one can associate various quantitative semantics
	to a metric-temporal-logic formula, depending on the filter's 
	kernel
	used: square, rounded-square, Gaussian, low-pass, band-pass, or
	high-pass. This remarkable connection between filtering and metric
	temporal logic allows us to freely navigate between the two, and 
	to
	regard signal-feature detection as logical inference.
	To the best of our knowledge, this connection has not been
	established before. We prove that our qualitative, filtering
	semantics is identical to the classical MTL semantics. We also
	provide a quantitative semantics for MTL, which measures the
	normalized, maximum number of times a formula is satisfied within
	its associated kernel, by a given signal. We show that this
	semantics is sound, in the sense that, if its measure is 0, then 
	the
	formula is not satisfied, and it is satisfied otherwise. We have
	implemented both of our semantics in Matlab, and illustrate their
	properties on various formulas and signals, by plotting their
	computed measures.
	
\end{abstract}

%
% The code below should be generated by the tool at
% http://dl.acm.org/ccs.cfm
% Please copy and paste the code instead of the example below. 
%
%\begin{CCSXML}
%	<ccs2012>
%	<concept>
%	<concept_id>10003752.10003790.10003793</concept_id>
%	<concept_desc>Theory of computation~Modal and temporal 
%	logics</concept_desc>
%	<concept_significance>500</concept_significance>
%	</concept>
%	<concept>
%	<concept_id>10011007.10011006.10011039</concept_id>
%	<concept_desc>Software and its engineering~Formal language 
%	definitions</concept_desc>
%	<concept_significance>500</concept_significance>
%	</concept>
%	<concept>
%	<concept_id>10011007.10011006.10011039.10011040</concept_id>
%	<concept_desc>Software and its engineering~Syntax</concept_desc>
%	<concept_significance>500</concept_significance>
%	</concept>
%	<concept>
%	<concept_id>10011007.10011006.10011039.10011311</concept_id>
%	<concept_desc>Software and its 
%	engineering~Semantics</concept_desc>
%	<concept_significance>500</concept_significance>
%	</concept>
%	</ccs2012>
%\end{CCSXML}
%
%\ccsdesc[500]{Theory of computation~Modal and temporal logics}
%\ccsdesc[500]{Software and its engineering~Formal language 
%definitions}
%\ccsdesc[500]{Software and its engineering~Syntax}
%\ccsdesc[500]{Software and its engineering~Semantics}
%
%\printccsdesc

% We no longer use \terms command
%\terms{Theory}

%\keywords{LTI filtering; temporal logic; MTL; qualitative semantics; 
%quantitative semantics; correspondence; measure}

%\section{Introduction}
\section{Introduction}
\label{sec:Introduction}
Starting with natural sciences, such as, chemistry, physics, and
mathematics, and ending with the applied sciences, such as, mechanical
engineering, electrical engineering, and computer science, the
process of \emph{filtering} plays a central role.
In each of the above-mentioned domains, it takes a signal $u$ as
input, and it produces a signal $y$ as output, where the components of
$u$ satisfying some given property are removed.

For example, in chemistry, a filter $f$ may remove particular
molecules from a given solution. In optics (physics) and electrical
engineering, $f$ may remove particular frequencies of $u$, as in a
low-pass, high-pass or band-pass filter. Finally, in computer science
(eg.~in functional programming), $f$ may remove the elements in a list
that satisfy a predicate $p$.

If the relation between $u$ and $y$ is linear, that is, if $f(u)$ is a
linear function, the filter is called \emph{linear}, otherwise it is
called \emph{nonlinear}. Moreover, if the operation of $f$ depends
only on the values of $u$ and not on time, the filter is called
\emph{time invariant}. The most commonly used filters, in all of the
above areas, are the linear, time-invariant (LTI) filters.

Intuitively, an LTI filter operates by sweeping a \emph{kernel}
distribution $k(s)$ over the entire domain of the lagged input signal
$u(t{-}s)$, performing for every $s$ the multiplication
$u(t{-}s)k(s)$, and than summing up (or integrating) the results, in
order to obtain the value of the output signal $y(t)$ at time $t$. One
says that $y\,{=}\,u{*}k\,{=}\,k{*}u$ is the \emph{convolution} of $u$
and $k$.

If one interprets $(+,\times,0,1)$ over the \emph{field} of reals,
then filtering a discrete- or continuous-time \{0,1\}-valued signal
$u$ results in a signal $y$ ranging over the reals. One can speak in
this case about a \emph{quantitative semantics} of the LTI filter:
\begin{equation}
  y(n)\,{=}\sum\limits_{i\,{=}\,0}^{\infty} u(n-i)k(i),
  \quad
  y(t)\,{=}\int\limits_{0}^{\infty} u(t-s)k(s)\,ds.
\end{equation}
An important aspect of an LTI filter $f$ is that, the kernel $k$ is
the response $f(\delta(s))$ of the filter to an impulse $\delta$
placed at the origin. In discrete time, $\delta$ is the Kronecker
function, while in continuous time, $\delta$ is the Dirac
distribution.

If one interprets $(+,\times,0,1)$ over the ($\max$,$\min$,0,1)
\emph{idempotent dioid} \cite{GondGDS}, then filtering a discrete- or 
continuous-time
\{0,1\}-valued signal $u$, results in a \{0,1\}-valued signal $y$. In
this case, one can speak about a \emph{qualitative semantics} of the
LTI filter. The filter has in this case the following form:
\begin{equation}
y(t)\,{=}\sup\limits_{s\,{=}\,0}^{\infty} \min(u(t-s),\, k(s)).
  \label{eqn:lConv}
\end{equation}
This qualitative semantics can be readily mapped to either linear
temporal logic (LTL), or to \emph{metric temporal logic (MTL)}, by
using discrete time for LTL, discrete or continuous time for MTL, and
choosing a \emph{rectangular window} distribution $k$. In particular,
(\ref{eqn:lConv}) will represent the \emph{finally operator}. The
other LTL/MTL operators can also be defined, either by duality, or by
choosing appropriate kernels $k$.

To the best of our knowledge, this remarkable correspondence between 
LTI filters and MTL has not been established before.
This correspondence allows us to freely navigate between logic and
signal processing, and to understand signal analysis, such as, feature
detection, as a logical inference. It also allows us to equip MTL with
various quantitative semantics, depending on the LTI-filter-kernel $k$
used: From square, to rounded, or to even a Gaussian windows, or from
low-pass, to band-pass, or to even high-pass-filter kernel.

We prove that our qualitative, LTI-filtering semantics is identical to
the classical MTL semantics. We also provide:

\emph{A quantitative semantics for MTL, which measures the normalized,
  maximum number of times a formula is satisfied within its associated
  kernel, with respect to a given signal}.

We show that this quantitative semantics is sound, in the sense that,
if its measure is 0, then the formula is not satisfied, and it is
satisfied otherwise. We implemented both of our semantics in Matlab,
and illustrate their properties on various formulas and signals by
plotting their measures.

It is important to note that an LTI filter is called \emph{causal} if
the value of $u$ is known for the entire kernel $k$ at the time the
multiplication is performed, and otherwise \emph{not-causal}. This can
be easily accomplished for signals with bounded domain. This
restriction corresponds to bounded MTL, too. In the rest of the paper
we will stick to this restriction.

The rest of the paper is organized as follows. In
Section~\ref{sec:related_work} we discuss related work. In
Section~\ref{sec:mtl} we revise  MTL, and in Section~\ref{sec:filt}
we revise LTI filters. In Section~\ref{sec:mtlAsFiltering} we show
that MTL semantics corresponds to a qualitative semantics in terms of
LTI filters, and then also provide associated quantitative semantics,
depending on the chosen kernel. Finally in
Section~\ref{sec:conclusions} we draw conclusions and discuss future
work.

%\section{Related work}
\section{Related Work}
\label{sec:related_work}
% !TEX root =  ../Rodionova.tex
Linear temporal logic (LTL)\cite{ltl}, and metric temporal logic (MTL)\cite{mtl}, have proven to be concise and elegant formalisms for
rigorously specifying a required temporal behaviour, for a system
under investigation. While in LTL real time is not of concern (one speaks
about logical time), in MTL all the logical formulas are qualified
with a time interval (window), during which the system has to satisfy
the corresponding formula. Moreover, while in LTL time is discrete, in
MTL time can be either discrete or continuous. In fact, by restricting
the windows in MTL to $[0,\infty]$, or the length of the signal, one
can obtain LTL or bounded LTL, respectively.

The classical semantics for LTL or MTL is \emph{qualitative}, that is,
it provides a true or false answer, to whether a signal satisfies, or
violates, a given LTL or MTL formula.  Moreover, the signals
themselves, are Boolean, too. With the increasing importance of
analog- or even mixed digital-analog signals, this classical semantics
has been extended to account for real-time, real-valued signals. To
this end, the set of atomic propositions has been extended to also
include relational propositions of the form $p\,{>}\,\theta$, where
$\theta$ is a constant value, representing a threshold. By negation and
conjunction, one obtains all other relational variations.

When monitoring LTL or MTL properties over real-time, real-valued,
possibly noisy signals, the classic, true or false interpretation, has
the limitation that, a small perturbation (also called a jitter), in
either the temporal- or in the value-domain, may affect the overall
verdict.  In the last decade, there has been a concerted effort, to
provide alternative ways to interpret LTL or MTL, by proposing a so
called \emph{quantitative}, or metric semantics for an LTL or MTL
specification. This is often referred to as a \emph{measure} of the
specification.

In~\cite{Rizk2011}, Rizk et.~al propose a quantitative semantics for
an LTL over the reals (LTL($\mathbb{R}$)) that computes, using a
constraint solving algorithm, the domain, that is, all the time points
in which the real variables occurring in a formula, make this formula
true, for a signal under analysis.

In~\cite{Fainekos2009}, Fainekos et~al.~introduce a notion of
\emph{space robustness} for MTL, with numeric predicates interpreted
over real-valued behaviors. The space robustness measures the degree
by which a continuous signal satisfies or violates the specification.
The key idea consists in computing for each moment of time the
distance $x_p(t)\,{-}\,\theta$, and in using min and max, to summarize
these distances. Consequently, min and max replace the Boolean
operators and, and or. The temporal operators globally and eventually,
had to be replaced accordingly, with inf and sup. This quantitative
interpretation of MTL over real-time, real-valued signals, has found a
number of applications in the hybrid systems community, for both the
identification of the uncertain
parameters~\cite{bartocci2015,Donze2010a,Donze2010b} in a system, and
for the falsification analysis of a model~\cite{sankaranarayanan2011}.

In~\cite{Donze2010}, Donz\'e et al.~define a notion of \emph{time
  robustness} for MTL interpreted over dense-real-time, real-valued
signals. The time robustness of a propositional formula 
evaluated at time $t$ measures the distance between $t$ and the 
nearest instant $t'$ in which the proposition switches its value. 
The other logical and temporal operators are computed using the 
same rules as in the space robustness.
In contrast, as shown in Section~\ref{sec:mtlAsFiltering}, we
provide within the same filtering framework, both a qualitative and a
quantitative semantics for MTL interpreted over real-time,
\emph{Boolean-valued} signals. Moreover, our quantitative semantics is
different from time robustness. The latter does not allow to integrate
the truth value of the formulas over time, and hence reason about the
satisfaction rate of a specification.
%The key idea in extending the qualitative semantics of the logic to this quantitative setting consists  
%in replacing the boolean operators $\vee$ and $\wedge$ with 
%\emph{min} and \emph{max} and the temporal operators using 
% \emph{inf, sup} operations. 
% It provides a measure of the time-shifting that the system can tolerate.
%None of the aforementioned works allows to integrate the truth values of the 
%formulas over time.  

Other notions of robustness based on accumulation have been in part
explored in the discrete time setting, motivated by the goal to equip
formal verification and synthesis with quantitative objectives.  For
example, in~\cite{Boker2014}, Boker et al. investigate the extension
of LTL with \emph{prefix-accumulation}, \emph{path-accumulation},
\emph{average} and \emph{infinite average} assertions.  However, the
interpretation of this extended logic remains in the Boolean domain.
Another related approach consists in extending temporal logics with
discounting operators~\cite{Alfaro2004,Almagor2015} weighting the
importance of the events according to how late they occur.  In the
context of real-time, the Duration Calculus (DC)~\cite{dc} is an
interval logic containing the special duration operator. This operator
allows to integrate the truth values of state expressions over an
interval of time.  Recently, Akazaki et al.~\cite{Akazaki2015} extend
MTL with \emph{average} temporal operators that quantifies the average
over time of the space and time robustness previously introduced. The
aforementioned logics based on the idea of accumulation share a common
aspect -- they are all equipped with special syntactic constructs that
allow to specify integration of values. This is in contrast to our
work, in which we use convolution to provide the power of value
accumulation directly to the semantics of MTL without the need to
adapt its syntax. Finally, we also note that we are not aware of any
previous studies that relate convolution and filtering to the
specification formalisms described in this section.

\begin{figure}[t]
%\vspace*{-3ex}
 \centering
 \epsfig{file=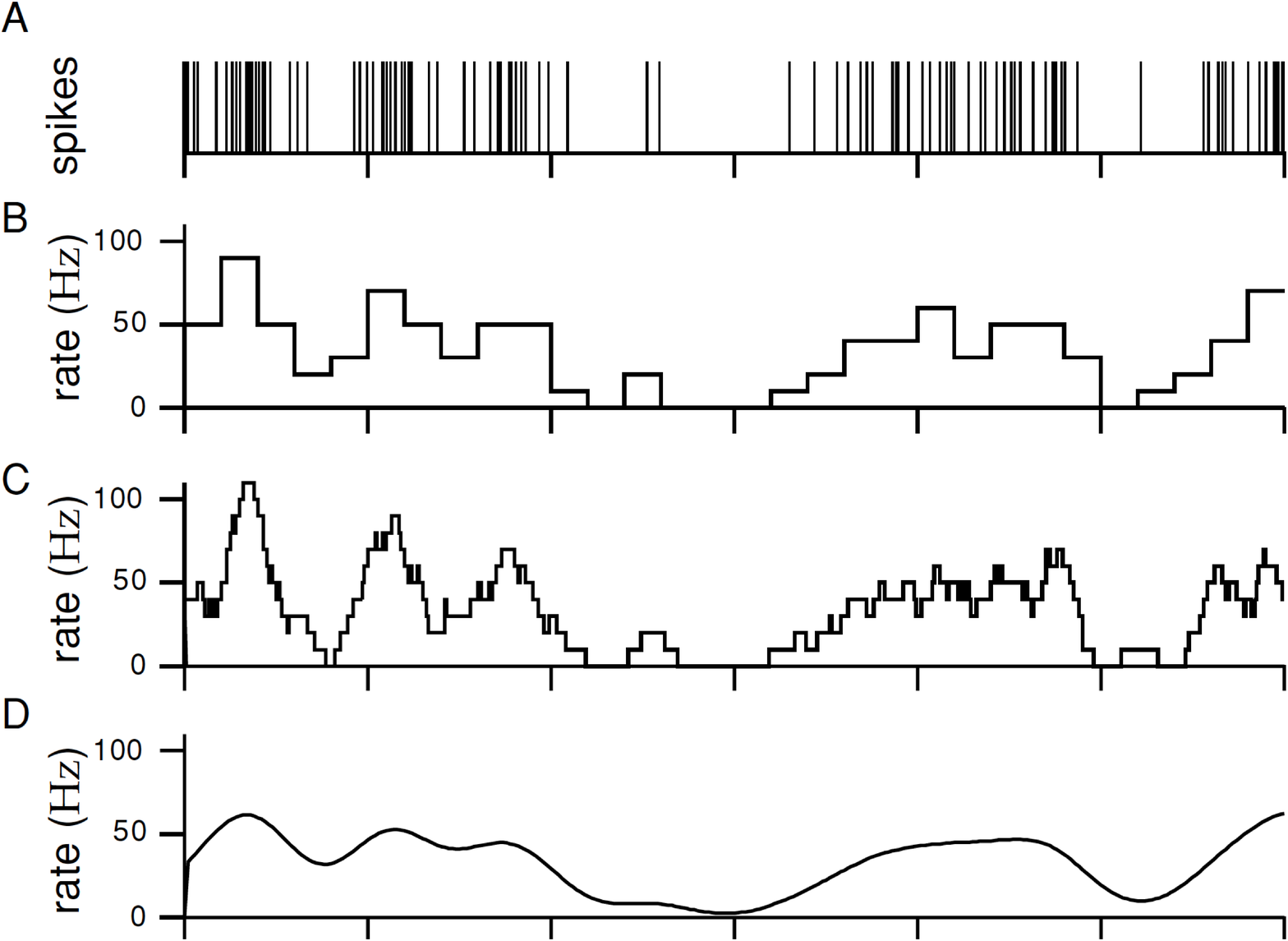, width=\linewidth}
 \caption{Measuring the spiking rate of a neuron.} 
 \label{fig:measuring-spike-rate}
%\vspace*{-2ex}
\end{figure}
%%

%\begin{itemize}
%%\item Time Frequency Logic~\cite{Donze2012}
%\item The Work of Alfaro~\cite{Alfaro2009}
%\item Durational Semantics 
%\item AverageSTL \cite{Akazaki2015}
%\end{itemize}

%\section{Metric Interval Temporal Logic}
\section{Metric Temporal Logic}
\label{sec:mtl}

Let $P\,{=}\,\{ p_{1},\ldots,p_{n}\}$ be a set of atomic
propositions. A {\em signal}
$x\,{:}\,\mathbb{T}\,{\rightarrow}\,2^{P}$ over $P$ is a set-valued
function, with time domain $\mathbb{T}$.  In this paper, we only
consider finite-length signals, that is, we restrict their domains to
$\mathbb{T}\,{=}\,[0,T]$, for an arbitrary, finite value $T$. In other
  words, $\mathbb{T}$ is an interval which is either a subset of
  non-negative reals $\mathbb{R}_{\geq 0}$ or naturals
  $\mathbb{N}$. We denote by $x_{p}$ the projection of $x$ to the
  proposition $p$.

Metric Temporal Logic (MTL)~\cite{mtl} is a real-time extension of
Linear Temporal Logic (LTL)~\cite{ltl}.  We consider a bounded variant
of MTL that contains both {\em future} and {\em past} temporal
operators. Its principal modalities are timed until $\until_{I}$ and
since $\since_{I}$, where $I$ is a non-empty interval.

We provide a generic definition of MTL which is consistent with both
the dense and discrete time interpretation of signals. Formally, the
syntax of MTL interpreted over such signals is defined by the
following grammar:
\begin{equation*}
   \varphi := p \mid \neg \varphi \mid \varphi_1 \vee \varphi_2 \mid
  \varphi_1 U_I \varphi_2 \mid \varphi_1 S_I \varphi_2,
\end{equation*}
where $p\,{\in}\,P$ and $I$ is a non-empty interval of the form
$\langle{a, b}\rangle$, such that, the left boundary $\langle$ is
either open ( or closed [, the right boundary $\rangle$ is either open
  ) or closed ], and the boundary values $a,b\in\mathbb{N}$ are
natural numbers with $0\,{\leq}\,a\,{\leq}\,b$.

The satisfaction of a given formula $\varphi$ with respect to a signal
$x$ at time point $i$ is a relation denoted by $(x,i) \models \varphi$
and defined inductively as follows:
\begin{align*}
%---------------------------------------------------------------------
&(x,i) \models p &&\iff&& x_p[i]=1 \\
%---------------------------------------------------------------------
&(x,i) \models \neg \varphi &&\iff&& (x,i) \not \models 
\varphi \\
%---------------------------------------------------------------------
&(x,i) \models \varphi \vee \psi &&\iff&& (x,i) \models 
\varphi \; \textrm{or} \; (x,i) \models \psi \\
%---------------------------------------------------------------------
&(x,i) \models \varphi \until_I \psi &&\iff&& \exists j \in 
(i + I) \cap \mathbb{T} \; \textrm{:} \;
(x,j) \models \psi \; \\ &&& && \textrm{and} \; 
\forall k\in(i, \,j),\ (x,k) \models \varphi \\
%\forall i < k < j, (x,k) \models \varphi \\
%---------------------------------------------------------------------
&(x,i) \models \varphi \since_I \psi &&\iff&& \exists j \in 
(i - I) \cap \mathbb{T} \; \textrm{:} \;
(x,j) \models \psi \; \\ &&& && \textrm{and} \; 
\forall k\in(j,\, i),\ (x,k) \models \varphi.
%\forall j < k < i, (x,k) \models \varphi.
%---------------------------------------------------------------------
\end{align*}
From the basic definition of MTL, we can derive other standard Boolean
and temporal operators as follows:
$$
\begin{array}{l}
  \top = p \vee \neg p,\quad  \bot =\neg \top,\quad
  \varphi \wedge \psi = \neg(\neg \varphi \vee \neg \psi),\\[1mm]
  \F_I \varphi = \top \until_I \varphi,\quad
  \G_I \varphi = \neg \F_I \neg \varphi,\\[1mm]
  \opP_I \varphi = \top \since_I \varphi,\quad
  \opH_I \varphi = \neg \opP_I \neg \varphi.\\[1mm]
\end{array}
$$
The Finally $\F_{I}$, Globally $\G_{I}$, Once $\opP_{I}$ and
Historically $\opH_{I}$ also admit a natural direct definition of
their semantics:
$$
\begin{array}{lcl}
(x,i) \models \F_{I} \varphi &\iff& \exists j \in (i + I) \cap 
\mathbb{T} \; \textrm{:} \;
(x,j) \models \varphi \\[1mm] 
(x,i) \models \G_{I} \varphi &\iff& \forall j \in (i + I) \cap 
\mathbb{T}, 
(x,j) \models \varphi \\[1mm] 
(x,i) \models \opP_{I} \varphi &\iff& \exists j \in (i - I) \cap 
\mathbb{T} \; \textrm{:} \;
(x,j) \models \varphi \\[1mm] 
(x,i) \models \opH_{I} \varphi &\iff& \forall j \in (i - I) \cap 
\mathbb{T}, 
(x,j) \models \varphi. \\ 
\end{array}
$$

%\section{Filtering}
\section{Linear Time-Invariant Filters}
\label{sec:filt}

One of the most fundamental operations in signal processing is
\emph{linear filtering}, that is, the \emph{convolution} of a signal
with an appropriately chosen \emph{kernel} or \emph{window}. From
audio, to image, and to video processing, linear filtering is used to
either transform a signal in an appropriate way (eg.~blur, sharpen) or
to detect its features (eg.~edges, patterns).

In order to understand linear filtering, it is instructive to describe
it first in terms of the most fundamental windowing primitives: the
Kronecker $\delta$ function, for the discrete-time case, and the
Dirac $\delta$ distribution, for the continuous-time case. This 
distribution is defined as follows:
\begin{equation}
  \delta(n)\,{=}\,\begin{cases}
    \infty, & \mbox{if } n=0\\ 
    0, & \mbox{otherwise},
  \end{cases}
  \qquad
  \int\limits_{-\infty}^{+\infty} \delta(s)\,ds\,{=}\,1.
\end{equation}
It is interesting to note, that Kronecker $\delta$ is a function,
whereas Dirac $\delta$ is not, that is, it makes sense only within an
integral. Starting with Cauchy and up to Dirac, it took a
long time to properly define this distribution ~\cite{diracD, 
cauchyD}.

The main motivation for developing the impulse windowing primitive,
was the ability to extract the instantaneous action of a function at a
particular moment of time. For example, the impulse provided by a
baseball bat when hitting the ball. This fundamental ability, allows
to describe any discrete (or continuous) signal $x$ as an infinite
summation (or integral):
\begin{equation}
  x(n)\,{=}\sum\limits_{i\,{=}\,-\infty}^{\infty} x(n{-}i)\delta(i),
  \quad
  x(t)\,{=}\int\limits_{-\infty}^{+\infty} x(t{-}s)\delta(s)\,ds.
\end{equation}
This infinite sum (or integral) is denoted as $x\,{*}\,\delta$ and it
is called the \emph{convolution} of $x$ and $\delta$. Convolution is
commutative, that is, $x{*}y\,{=}\,y{*}x$. Since $x\,{=}\,x{*}\delta$,
convoluting $x$ with $\delta$ is an identity transformation. Its main
appeal, is that it allows to express any linear, time-invariant (LTI)
function $f$, as the convolution of its response $y$ to an impulse
$\delta$ at the origin, with its (lagged) input $u$.  For the
discrete-time case (the continuous-time case is analogue), one can
write:
\begin{equation}
  f(\sum\limits_{i\,{=}\,-\infty}^{\infty} u(n{-}i)\delta(i))\,{=}\,
  \sum\limits_{i\,{=}\,-\infty}^{\infty} u(n{-}i)f(\delta(i)),
\end{equation}
where $y(i)\,{=}\,f(\delta(i))$ is the response of $f$ to $\delta$ at
time zero. This response can be seen as the filtering kernel, and
therefore $f$, as a linear, time-invariant filter. The popularity of
such filters stems from the fact that the response of $f$ to an
impulses $\delta$ at the origin is easy to determine experimentally.

\begin{figure}
 \centering
 \epsfig{file=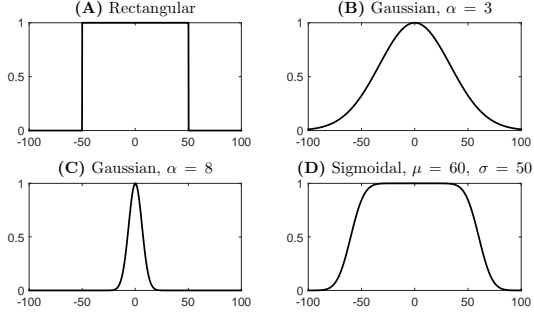, width=\linewidth}
 \caption{Windowing function types}
 \label{windfun}
\end{figure}
In order to motivate the use of various kernels (windows), and to
establish a link between LTI-filtering and the semantics of MTL's
\emph{finally operator}, we borrow an example from~\cite{abottBook}:
measuring the firing rate of a spiking neuron.

If one ignores the duration and the pulse-like shape of a spike (also
called an action potential), a spike sequence $\rho$ can be simply
characterized as a sum of Dirac distributions:
\begin{equation}
  \rho(t)\,{=}\,\sum\limits_{i\,{=}\,0}^{n} \delta(t{-}t_i),
  \qquad \quad
  r\,{=}\,\frac{1}{T}\int\limits_{0}^{T} \rho(s)\,ds,
\end{equation}
where $t_i$ is the time where the neuron generates a spike. The
normalized integral $r$ of $\rho(t)$ over an interval [0,T], is
called the \emph{spike-count rate} for that interval. 
%If $\rho$ is averaged and $T$ goes to zero, then $r(t)$ is 
%called the \emph{spike rate}.

Now suppose one would like to approximate the spike rate of a neuron
from an experimentally-obtained spike train. For example,
Figure~\ref{fig:measuring-spike-rate}(A) shows the spike train
obtained from a monkey's inferotemporal-cortex neuron, while watching
a movie.  There are various ways to computing this rate.

A very simple-minded way is to divide the signal's time domain in a
set of disjoint bins, say of length $\Delta{t}$, count the number of
spikes in each bin, and divide by $\Delta{t}$. For example,
Figure~\ref{fig:measuring-spike-rate}(B) shows this approach, for
$\Delta{t}\,{=}\,100$ms. The result is a discrete (in multiples of
$1/\Delta{t}$), piecewise constant, function. Decreasing $\Delta{t}$
increases the temporal resolution, at the expense of the rate
resolution. Moreover, the bin placement has an impact on the computed
rates, too.
%
% it is like in FFT - resolution in  frequency is 1/T
% Example. let sampling frequency be 1, so there is no possiblity to 
%%% have 2 spikes in one time interval. 
% So let T be 50, then possible values can be 0, 1, ... 50 spikes and 
%%% resolution is 1/50. 
% If we have T=2, then just 0, 1, 2 and resolution is 1/2, which is a
% bigger step, so this rate resolution is worse.
%
% http://zabbarob.github.io/measuring-firing-rates-webpage/
% If we count over the whole trial we get the total number of spikes 
%but loose all temporal resolution.
%To know how the spike count changes over time we can subdivide the 
%%%trial into smaller time bins.
%The smaller the time bins the better the temporal resolution. But 
%%%this comes at the expense of a decreasing rate resolution. For 
%%%very 
%%%%small time bins the rate can only be either 1 or 0.
%

One can avoid an arbitrary bin-placement, by taking only one bin (or
window) of duration $\Delta{t}$, slide it along the spike train, and
then counting for each position $t$, the number of spikes within the
bin. This approach is shown in
Figure~\ref{fig:measuring-spike-rate}(C), where the window size is
$\Delta{t}\,{=}\,100$ms. A centered discrete-time window can be
defined as follows:
\begin{equation}
  w_{\Delta{t}}(n)\,{=}\,
  \frac{1}{\Delta{t}}\sum\limits_{i\,{=}\,-\Delta{t}/2}^{+\Delta{t}/2} \delta(n{-}i).
\end{equation}

A centered continuous-time window of size $\Delta{t}$ can be
defined as follows (this would work in discrete-time, too):
 \vspace*{+1.5ex}
\begin{equation}
  w_{\Delta{t}}(t)\,{=}\,\begin{cases}
     1/\Delta{t}, & \mbox{if } -\Delta{t}/2\,{\leq}\,t\,{\leq}\,\Delta{t}/2\\ 
     0, & \mbox{otherwise}.
  \end{cases}  
\end{equation}

\begin{figure}[t]
\centering
\epsfig{file=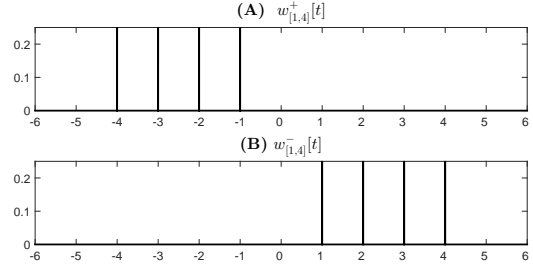, width=\linewidth}
\vspace*{-5ex}
\caption{Rectangular windows  $w^+_{[1,4]}[t]$ and $w^-_{[1,4]}[t]$}
\label{rectwin}
%\vspace*{-3ex}
\end{figure}
 \vspace*{+1.5ex}
The bin sliding along the spike train, the counting, and the
normalization with $\Delta{t}$, can be succinctly expressed in terms 
of
convolution, either in discrete or continuous time:
\begin{equation}
  r(n)\,{=}\sum\limits_{i\,{=}\,0}^{T} \rho(i)w_{\Delta{t}}(n{-}i),
  \quad
  r(t)\,{=}\int\limits_{0}^{T} \rho(s)w_{\Delta{t}}(t{-}s)\,ds,
  \label{eq:spike-rate}
\end{equation}
where $[0,T]$ is the domain of the spike train $\rho(t)$. If $\rho(t)$
is defined to be 0 outside $[0,T]$ one can take the sum (or the
integral) over $[-\infty,+\infty]$.  The convolution equation defines
a discrete (and continuous) \emph{linear filter} with \emph{kernel}
$w_{\Delta{t}}$. Note that the rates $r(t)$ at times less than one bin
apart are correlated, because they involve common spikes. Moreover,
the discontinuous nature of $r(t)$ is caused by the discontinuous
nature of the window $w_{\Delta{t}}$ at its borders.

One can smoothen $r(t)$ and reduce the correlation between bins, by
using a Gaussian window $N(0,\sigma)$ with $\sigma\,{=}\,100$ms, as
shown in Figure~\ref{fig:measuring-spike-rate}(D). A Gaussian is
defined as follows:
\begin{equation}
  N(\mu,\sigma)(t)\,{=}\,
  \frac{1}{\sigma\sqrt{2\pi}}e^{-\frac{(t-\mu)^2}{2\sigma^2}}.
\end{equation}
In fact, one could use any windowing function, as long as the area
below the window (the integral) is one.

The spike train $\rho(t)$ discussed above could as well represent the
truth-values of a discrete-time MTL signal $x$. Alternatively,
connecting equal, successive values in $\rho(t)$, would result in the
truth-values of a continuous-time MTL signal $x$. Hence, the
LTI-filters above can be understood as a \emph{quantitative semantics}
for the MTL operators finally and once interpreted over
$x$: \emph{The percentage of satisfaction of $x$ within the associated
  window}. The size and the shape of the windows allow (as shown in
the next section) to introduce a \emph{temporal jitter}
with respect to a satisfaction. Moreover, considering only square windows, and
interpreting $(+,\times,0,1)$ over the \emph{idempotent dioid}
$(\max,\,\min,\,0,\,1)$ gives a logical, \emph{qualitative} semantics to the
LTI-filter as the classic finally operator.

%\section{MTL operators as convolution}
%\vspace*{+2ex}
\section{MTL as LTI-Filtering}
\label{sec:mtlAsFiltering}
\begin{figure}[tp]
	\vspace*{-3ex}
	\centering
	\epsfig{file=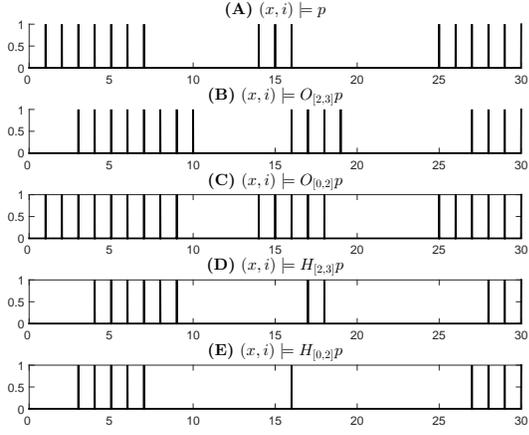, width=\linewidth}
	\vspace*{-6ex}
	\caption{Qualitative discrete-time semantics.}
	%  \vspace*{-2ex}
	\label{fig:qualitative_FGd}
\end{figure}

As it has already been mentioned before, all LTI-filters are
characterized by the particular type of the kernel (window).  Although
there are a lot of such kernels, the main idea is the same: a kernel
function should have a bounded support, that is, a bounded domain
where the function is not zero.

General definitions allow such kernels to go sufficiently rapid
towards zero. The simplest finite support window is when the function
is constant inside of some interval, and zero elsewhere. This function
is called a \emph{rectangular window} or ``Boxcar'' (see
Figure~\ref{windfun}). A smooth Gaussian window is the second type of
a window, because it extends to infinity. So it should be truncated at
the ends of the window. The main advantage of such a function is its 
smooth nature.

Let us start with a rectangular window definition:
%
%\newdef{definition}{Definition}
%\begin{definition}
   for a given continuous-time MTL formula with constrained time
   interval $I\,{=}\,\langle{a},b\rangle$, a rectangular window
   function $w_I(t) :\mathbb{R} \rightarrow [0,1]$ is constructed as
   follows:
   \begin{equation}
    w^{+}_I(t)\,{=}\,\begin{cases}
    \frac{1}{|I|}, & \mbox{if } t\in \tilde{I} \\
    0,  & \mbox{if } t\not\in \tilde{I},
    \end{cases}
    \quad
    w^{-}_{I}(t)\,{=}\,\begin{cases}
             \frac{1}{|I|}, & \mbox{if } t\in I \\
             0,  & \mbox{if } t\not\in I,
    \end{cases}
   \end{equation}
   where $\tilde{I}\,{=}\,\langle{}{-}b, {-a}\rangle{}$ is used for
   the future MTL operators $\F_{I}$ or $\G_{I}$ and
   $I\,{=}\,\langle{}a,b\rangle{}$ for the past MTL
   operators $\opP_I$ or $\opH_I$.
   If $I\,{=}\,[a,a]$ is singular, then $w_{[a,a]}$ is defined as
   the Dirac-$\delta$ distribution. We will denote
   this window by $w_{\{a\}}$.
      
   For discrete-time, open intervals can always be replaced with
   closed intervals. Hence we restrict $I$ to closed intervals,
   only. In this case, we define $w_I$ in terms of Kronecker-$\delta$,
   as follows:
\begin{equation}
 w^+_I[t]\,{=}\,\frac{1}{|I|}\sum\limits_{i\,{\in}\,I} \delta(t+i), \quad   
 w^-_I[t]\,{=}\,\frac{1}{|I|}\sum\limits_{i\,{\in}\,I} \delta(t-i),\\
\label{def:win}
 \end{equation}
where the ${+}$ and the ${-}$ superscripts correspond to the future and 
past MTL operators, respectively.
%\end{definition}

% As before, ${+}$ and ${-}$ are the signs in corresponding
% impulse-train kernels, with regard to a zero point. 
Note that for past MTL operators one needs to delay the kernel, so
the sign is negative, and for future MTL operators one has to rush the
kernel, so the sign is positive. For instance, consider the
discrete-time MTL formulas $\F_{[1,4]}\,p$ and
$\opP_{[1,4]}\,p$. The windows used by these formulas:
\begin{equation}
 w^+_{[1,4]}[t]\,=\,\frac{1}{4}\sum\limits_{i=1}^{4}\delta(t+i),\quad
 w^-_{[1,4]}[t]\,=\,\frac{1}{4}\sum\limits_{i=1}^{4}\delta(t-i),
\end{equation}
are shown graphically in Figure~\ref{rectwin}(A) and (B), respectively.

\newdef{remark}{Remark} \newdef{exmp}{Example}

%\section{Qualitative semantics}
\subsection{Qualitative Semantics}
\label{sec:qualit}
 \begin{figure}[tp]
 	%  \vspace*{-3ex}
 	\centering
 	\epsfig{file=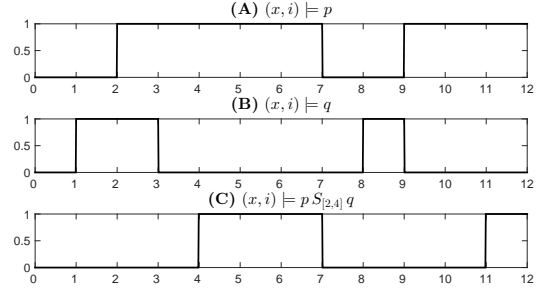, width=\linewidth}
 	\vspace*{-4ex}
 	\caption{Continuous-time semantics for $p\since_{[2,4]}{q}$}
 	\label{fig:qualitative_Sc}
 	%  \vspace*{-2ex}
 \end{figure}
In this section, we show that all temporal operators of MTL can be
defined in terms of LTI-filtering. For this purpose, we interpret
addition, multiplication, and their associated neutral elements over
$(\{0,1\}, \max, \min, 0, 1)$ the \emph{max-min idempotent dioid} 
with value
domain \{0,1\}. In this dioid one can define the \emph{complement} 
$\neg{v}$ of a value $v$ as $1{-}v$.
Moreover, the normalization factor in the definition of the window 
function  
$w_I(t)$ should be omitted in order to preserve Boolean values. 
%The $\delta$ distribution is interpreted in both the discrete- and 
%continuous-time semantics as the
%Kronecker-$\delta$ function. However, the convolution integral is 
%going
%to be interpreted differently.
%%%%%%%%%%%%%%%%%%%%%%%%%%%%%%%%%%%%%%%%%%%%%%%%%%%%%%%%%%%%%%%%%%%%%%
\subsubsection{Discrete-time semantics}

In the discrete-time domain, any signal is bounded (remember that we
only consider boun\-ded MTL), and therefore all the intervals of the 
MTL
operators are bounded, too. As a consequence, the number of discrete
points in each interval is finite, and we can extend the binary 
$\max$
operator to an \emph{n-ary max operator}, which we can use for the
convolution integral. Thus, the discrete-time semantics of MTL can be
formulated as below:
 \begin{align*}
%%--------------------------------------------------------------------
 % Proposition
 &(x,i) \models p && \iff && x_p[i] \\
%%--------------------------------------------------------------------
 % Negation
 &(x,i) \models \neg \varphi && \iff && 
 1 - (\,(x,i) \models\varphi\,)\\
%%--------------------------------------------------------------------
 % Or
 &(x,i) \models \varphi \vee \psi && \iff &&
 \max\left((x,i) \models \varphi,\,(x,i) \models \psi\right)
 \\ %[1mm]
%%--------------------------------------------------------------------
 % Finally
 &(x,i) \models {\F}_{I} \varphi && \iff &&
 \max_{j\,\in\,\mathbb{T}} \min\left((x,j)\models\varphi, \,
 w^{+}_I[i\,{-}\,j]\right) \\
%%--------------------------------------------------------------------
 % Globally
 &(x,i) \models {\G}_{I} \varphi && \iff &&
 1 - (\,(x,i) \models {\F}_{I} \neg \varphi\,)
 \\ %[1mm] 
%%--------------------------------------------------------------------
 % Once
 &(x,i) \models {\opP}_{I} \varphi && \iff &&
 \max_{j\,\in\,\mathbb{T}} \min\left((x,j)\models\varphi,\,
 w^{-}_I[i\,{-}\,j]\right)\\
%%--------------------------------------------------------------------
 % Historically
 &(x,i) \models {\opH}_{I} \varphi && \iff &&
 1 - (\,(x,i) \models {\opP}_{I} \neg \varphi\,)
 \\%[1mm] 
%%--------------------------------------------------------------------
 % Until
 &(x,i) \models \varphi \until_I \psi && \iff && 
 \max_{j\,\in\,I}\min\Big(
 (x, i)\models {\G}_{[1,j-1]}\varphi,\,\\
 &&&&&
  (x, i)\models {\F}_{\{j\}}\psi\Big) \\
%%--------------------------------------------------------------------
 % Since
 &(x,i) \models \varphi \since_I \psi && \iff &&
 \max_{j\,\in\,I}\min\Big(
 (x, i)\models {\opH}_{[1,j{-}1]}\varphi,\, \\
 &&&&&
 (x, i)\models {\opP}_{\{j\}}\psi\Big).
%%--------------------------------------------------------------------
 \end{align*}
In fact, we can define the entire semantics of MTL 
in terms of LTI-filtering, by using the Kronecker-$\delta$ kernel:

\begin{align*}
  &(x,i) \models p && \iff &&
  \max_{j\,\in\,\mathbb{T}} \min\left(x_p[j], \, \delta(i-j)\right)\\
  &(x,i) \models \neg \varphi &&\iff &&
  \max_{j\,\in\,\mathbb{T}}\min\left(1 - ((x,j) \models 
  \varphi),\delta(i-j)\right) \\
  &(x,i) \models \varphi \vee \psi && \iff &&
  \max_{j\,\in\,\mathbb{T}} \min\Big(\max\big(\\
  &&&&&
  (x,j) \models \varphi,\,(x,j) \models\psi\big),\,\delta(i-j)\Big).
\end{align*}

In order, to illustrate the qualitative, discrete-time
semantics for various temporal operators, consider the example shown
in Figure~\ref{fig:qualitative_FGd}. In
Figure~\ref{fig:qualitative_FGd}(A) we show the values $x_p$ of
discrete-time signal $x$, for which it satisfies proposition $p$.
This signal is defined over the interval $[0, 30]$, with
$\Delta{t}\,{=}\,1$. In other words, the domain of $x$ is
$\{0,1,\ldots,30\}$.

In order to interpret the MTL operators $\opP$ and $\opH$ over $x$, 
we
use two indexing windows $I$, for both: The first is $I=[2,3]$, and 
the
second is $I=[0, 2]$. To simplify the caption in the Matlab figures,
we use the textual representation of the temporal operators: F
(finally) for $\F$, G (globally) for $\G$, O (once) for $\opP$, and H
(historically) for $\opH$.

The satisfaction of the MTL formulas $\opP_{[2,3]}p$ and
$\opP_{[0,2]}p$ for the signal $x_p$ are shown in
Figure~\ref{fig:qualitative_FGd}(B-C). Similarly, the satisfaction of
$\opH_{[2,3]}p$ and $\opH_{[0,2]}p$ for the signal $x_p$ are shown
in Figure~\ref{fig:qualitative_FGd}(D-E). The dependence on the 
length
of the interval $I$ is precise:  Increasing the length, increases the
number of points where the operator $\opP_{I}p$ is satisfied and
$\opH_{I}p$ is not. Moreover, the shift of the interval's beginning
point leads to the same shift of the output signal.

For the qualitative, discrete-time semantics of an MTL formula
$\varphi$ with respect to a signal $x$, it is straight forward to
show, that the semantics is sound, in the sense that, if it produces
the value 0, then the formula is not satisfied in the classical
semantics, and if it produces 1 then it is satisfied.  Let us index
satisfaction with $C$ as classic semantics, and with $L$ as
LTI-filtering semantics in the following theorem.
\newtheorem{theorem}{Theorem}
\begin{theorem}[Soundness]
	For an MTL formula $\varphi$ and a 
	discrete-time signal $x$, it holds that $x$ satisfies $\varphi$ 
	in the LTI-filtering interpretation, if and only if, it satisfies 
	$\varphi$ in the classic interpretation. More formally one can 
	write:
	\vspace*{+2ex}
  \begin{align*}
  (x,i) \models_L \varphi &\iff
  (x,i) \models_C \varphi.
  \end{align*}
  \label{thm:qualiDTs}
\end{theorem}
\vspace*{-5ex}
\begin{proof} (Sketch)
  First, observe that the Boolean operators are interpreted the
  same way, as the max-min dioid provides a logical interpretation 
for
  the set \{0,1\}: the truth tables for conjunction and $\min$,
  disjunction and $\max$, and negation and complement, are all the 
same, respectively. Second, the existential quantifier in the classic
  interpretation, corresponds to max in the max-min dioid, which 
  leads to the same interpretation for $\F_I$ and $\opP_I$. Through 
  the properties of negation, we immediately establish the same
  interpretation for $\G_I$ and $\opH_I$. Finally, the definition of
  $\until_I$ and $\since_I$ can be written solely in terms of the
  previous temporal operators, as we did in the LTI-filter semantics.
\end{proof}

%%%%%%%%%%%%%%%%%%%%%%%%%%%%%%%%%%%%%%%%%%%%%%%%%%%%%%%%%%%%%%%%%%%%%%
\begin{figure}[tp]
	\vspace*{-3ex}
	\centering
	\epsfig{file=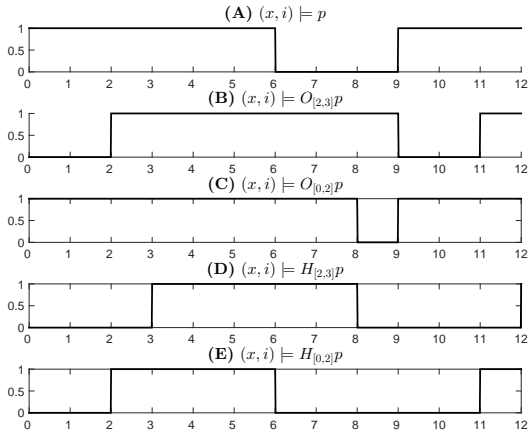, width=\linewidth}
	\vspace*{-6ex}
	\caption{Qualitative continuous-time semantics.}
	%  \vspace*{-2ex}
	\label{fig:qualitative_FGc}
\end{figure}
\subsubsection{Continuous-time semantics}
\begin{figure}[t]
	\vspace*{-3ex}
	\centering
	\epsfig{file=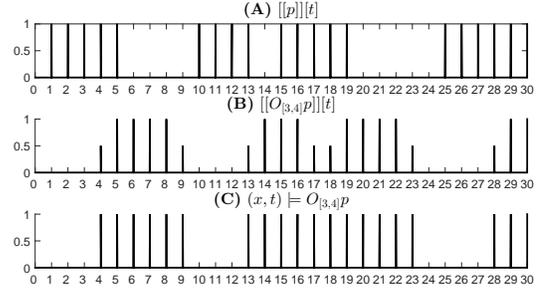, width=\linewidth}
	\vspace*{-6ex}
	\caption{Soundness of discrete-time semantics.}
	\vspace*{-2ex}
	\label{fig:thm_d}
\end{figure}
In the continuous-time semantics, each bounded interval is going to
have an infinite number of points, except for the singular intervals,
which contain only a single point. As a consequence, we have to
interpret the integral in this case, as the \emph{supremum operator},
which is the proper extension of $\max$ over infinite domains.
 \begin{align*}
%%--------------------------------------------------------------------
 % Proposition
 &(x,i) \models p && \iff && x_p(i) \\
%%--------------------------------------------------------------------
 % Negation
 &(x,i) \models \neg \varphi && \iff && 
 1 - (\,(x,i) \models\varphi\,)\\
%%--------------------------------------------------------------------
 % Or
 &(x,i) \models \varphi \vee \psi && \iff &&
 \max\left((x,i) \models \varphi,\,(x,i) \models \psi\right)
 \\ %[1mm]
%%--------------------------------------------------------------------
 % Finally
 &(x,i) \models {\F}_{I} \varphi && \iff &&
 \sup_{j\,\in\,\mathbb{T}} \min\left((x,j)\models\varphi, \,
 w^{+}_I(i\,{-}\,j)\right) \\
%%--------------------------------------------------------------------
 % Globally
 &(x,i) \models {\G}_{I} \varphi && \iff &&
 1 - (\,(x,i) \models {\F}_{I} \neg \varphi\,)
 \\
%%--------------------------------------------------------------------
 % Once
 &(x,i) \models {\opP}_{I} \varphi && \iff &&
 \sup_{j\,\in\,\mathbb{T}} \min\left((x,j)\models\varphi,\,
 w^{-}_I(i\,{-}\,j)\right)\\
%%--------------------------------------------------------------------
 % Historically
 &(x,i) \models {\opH}_{I} \varphi && \iff &&
 1 - (\,(x,i) \models {\opP}_{I} \neg \varphi\,)
 \\%[1mm] 
%         \end{align*}
%         \begin{align*}
%%--------------------------------------------------------------------
 % Until
 &(x,i) \models \varphi \until_I \psi && \iff && 
 \sup_{j\,\in\,I} \min\Big(
 (x, i)\models {\G}_{(0,j)}\varphi ,\, \\
 &&&&& 
 (x, i)\models {\F}_{\{j\}}\psi\Big)\\
%%--------------------------------------------------------------------
 % Since
 &(x,i) \models \varphi \since_I \psi && \iff &&
 \sup_{j\,\in\,I}\min\Big(
 (x, i)\models {\opH}_{(0,j)}\varphi,\, \\
 &&&&&
 (x, i)\models {\opP}_{\{j\}}\psi\Big).
%%--------------------------------------------------------------------
 \end{align*}
%As before, we can define the entire semantics of MTL in terms of
%LTI-filtering, by using the Kronecker-$\delta$ kernel:
%\begin{align*}
%%%--------------------------------------------------------------------
% &(x,i) \models p && \iff &&
% \sup_{j\,\in\,\mathbb{T}} \min\left(x_p(j),\, \delta(i-j)\right) \\
%%%--------------------------------------------------------------------
% &(x,i) \models \neg \varphi && \iff &&
% \sup_{j\,\in\,\mathbb{T}} \min\left(1 - ((x,j) \models 
% \varphi),\delta(i-j)\right) \\
%%%--------------------------------------------------------------------
% &(x,i) \models \varphi \vee \psi && \iff &&
% \sup_{j\,\in\,\mathbb{T}} \min\Big(\max\big(\\
% &&&&&
% (x,j) \models \varphi,\,(x,j) \models \psi\big),\,\delta(i-j)\Big).
%\end{align*}
In order to illustrate the qualitative, continuous-time
semantics for various temporal operators, consider the example shown
in Figure~\ref{fig:qualitative_FGc}. In
Figure~\ref{fig:qualitative_FGc}(A) we show $x_{p}$, the value of
continuous-time signal $x$, for which proposition $p$ is true.  This
signal is defined over the interval $[0, 12]$.  Every sub-interval in
the plot is left-closed and right-open.

In order to interpret the MTL operators $\opP$ and $\opH$ over $x$, 
we
use, the same indexing windows as before: First, $I=[2,3]$ and next,
$I=[0, 2]$. The satisfaction of the MTL formulas $\opP_{[2,3]}p$
and $\opP_{[0,2]}p$ for the signal $x$ are shown in
Figure~\ref{fig:qualitative_FGc}(B-C). Similarly, the satisfaction of
$\opH_{[2,3]}p$ and $\opH_{[0,2]}p$ for the signal $x$ are shown
in Figure~\ref{fig:qualitative_FGc}(D-E).

The qualitative, continuous-time semantics of the since operator
$(x,i)\models p \since_{[2,4]} q$ is shown in
Figure~\ref{fig:qualitative_Sc}. The semantics of $(x,i) \models p$
(signal $x_p$) is shown in Figure~\ref{fig:qualitative_Sc}(A) and
the semantics of $(x,i) \models q$ (signal $x_q$) shown in
Figure~\ref{fig:qualitative_Sc}(B). Finally, the semantics of $(x,i)
\models p \since_{[2,4]} q$ is shown in
Figure~\ref{fig:qualitative_Sc}(C). It should be emphasized that, the
obtained qualitative semantics corresponds to a standard MTL
semantics.
\begin{theorem}[Soundness]
	For an MTL formula $\varphi$ and a continuous-time signal $x$, it 
	holds that $x$ satisfies $\varphi$ in the LTI-filtering 
	interpretation, if and only if, it satisfies $\varphi$ in the 
	classic interpretation. More formally one can write:
   \begin{align*}
   (x,i) \models_L \varphi &\iff
   (x,i) \models_C \varphi.
   \end{align*}
\end{theorem}
%\vspace*{-2ex}
\begin{proof} (Sketch)
  The proof is very similar to the one for the discrete-time,
  qualitative semantics. The only difference is that we work now with
  windows having dense (infinite) domain, and max have to be 
therefore
  extended to sup, the proper extension on such domains.
\end{proof}

%\section{Quantitative Semantics}
\subsection{Quantitative Semantics}
	\begin{figure}[tp]
		\centering
		\epsfig{file=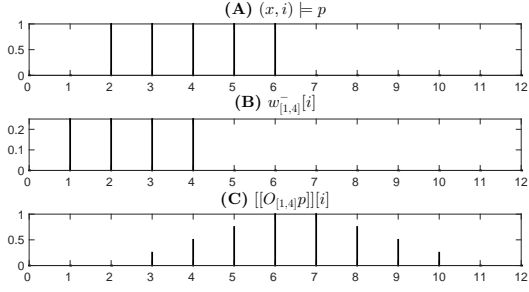, width=\linewidth}
		\caption{Quantitative discrete-time semantics.}
		\label{fig:quantitative_exd}
	\end{figure}

If we interpret addition, multiplication, and their associated 
neutral
elements over $([0,1],+,\times,0,1)$, the \emph{field of reals}
restricted to the interval [0,1], we recover the standard definition
of LTI-filtering. This definition allows us to associate a
\emph{quantitative semantics} to MTL temporal formulas.

This semantics measures the normalized, maximum number of times
the formula is satisfied within its associated window.
%\emph{This semantics measures the normalized, maximum number of times
%  the formula is satisfied within its associated window}.
Since $\G_I$ and $\opH_{I}$ can be satisfied only once within $I$, 
the
measure returns 1, in case of satisfaction, and 0, otherwise. 
However,
$\F_I$ and $\opP_{I}$ are more interesting, as they may be satisfied
several times in $I$. 
%Hence, they return a value greater than 0, in
%case of satisfaction, and 0 otherwise. 
%In order to avoid interference
%between Boolean and temporal MTL satisfaction, we use a max 
%semantics
%for disjunction and a min semantics for conjunction. 
As in Section~\ref{sec:qualit}, we distinguish between discrete- and
continuous-time semantics, and use summations and integrals,
respectively. 

Moreover, we give up the duality property between the $\G_I$ and 
$\F_I$ 
operators (same for past operators $\opH_{I}$ and $\opP_{I}$).
%Since any MTL formula %$\varphi$ 
%can be expressed in a semantically equivalent 
%\textit{positive normal form}, where negations may only 
%occur in front of atomic propositions \cite{joelDecMTL},
%in the rest of the paper we will stick to this representation form 
%of an MTL formula. 
In the rest of the paper we will stick to the \textit{positive normal 
form} representation of an MTL formula, where negations may only 
occur in front of atomic propositions \cite{joelDecMTL}.
%%%%%%%%%%%%%%%%%%%%%%%%%%%%%%%%%%%%%%%%%%%%%%%%%%%%%%%%%%%%%%%%%%%%%%
\subsubsection{Discrete-time semantics}
Assume we are given a discrete signal with $[0, T]$ duration
interval. Then the quantitative semantics for MTL can be formulated 
as
follows:
\begin{align*}
%---------------------------------------------------------------------
% Proposition 
&\val{ x,p } [i] && = && x_p[i] \\
%---------------------------------------------------------------------
% Negation 
&\val{ x,\neg p } [i] && = && 1 - \val{ x,p }[i]\\
%---------------------------------------------------------------------
% Or
&\val{ x, \varphi\,{\vee}\,\psi } [i] && = &&
\max(\val{ x,\varphi } [i],\ \val{ x,\psi }[i])\\
%---------------------------------------------------------------------
% And
&\val{ x, \varphi\,{\wedge}\,\psi } [i] && = &&
\min(\val{ x,\varphi } [i],\ \val{ x,\psi }[i])\\
%---------------------------------------------------------------------
% Finally
&\val{ x, {\F}_{I}\varphi } [i]&& = &&
\sum_{j\,{\in}\,\mathbb{T}} \ \val{ x,\varphi}[j]\cdot w^+_I[i-j]\\
%---------------------------------------------------------------------
% Once
&\val{ x, {\opP}_{I}\varphi } [i]&& = &&
\sum_{j\,{\in}\,\mathbb{T}} \ \val{ x,\varphi}[j]\cdot w^-_I[i-j]\\
%---------------------------------------------------------------------
% Globally
&\val{ x, {\G}_{I}\varphi } [i]&& = &&
\min_{j\,{\in}\,i+I} \val{ x,\varphi}[j]\\
%---------------------------------------------------------------------
% Historically
&\val{ x,{\opH}_{I}\varphi } [i]&& = &&
\min_{j\,{\in}\,i-I} \val{ x,\varphi}[j]\\
%\end{align*}
%\begin{align*}
%---------------------------------------------------------------------
% Until
&\val{ x,\varphi\until_I\psi }[i]&& = &&
\frac{1}{|I|}\sum_{j\,{\in}\,I} \ 
\val{ x,{\G}_{[1,j-1]}\varphi }[i]\cdot 
\val{ x,{\F}_{\{j\}}\psi    }[i] \\
%---------------------------------------------------------------------
% Since
&\val{ x,\varphi\since_I\psi} [i]&& = &&
\frac{1}{|I|}\sum_{j\,{\in}\,I} \ 
\val{ x,{\opH}_{[1,j-1]}\varphi }[i]\cdot 
\val{ x,{\opP}_{\{j\}}\psi    }[i],
\end{align*}
where the window superscripts $+$ and $-$ correspond to the window
direction, according to Equations~(\ref{def:win}). 
%Note that
%$\varphi[j]$ in the definition of globally and historically, could
%also have been written as the convolution $\varphi\,{*}\,\delta$. 
%As
%usual, the $\min$ operation over an empty set is assumed to be 1.

For example, consider the discrete-time signal $x$ over the interval
$[0, 12]$, with $\Delta{t}\,{=}\,1$, shown in
Figure~\ref{fig:quantitative_exd}. Since it is discrete-time, its
domain is $\{0,1,\ldots,12\}$.  In order to interpret the MTL 
operator
$\opP_{[1, 4]}\,p$, we construct the rectangular windowing function
$w^-_{[1, 4]}[i]$ as shown at the beginning of
Section~\ref{sec:mtlAsFiltering}.  The representation of this window
is shown in Figure ~\ref{rectwin}(B).  Since the MTL formula, is a
past formula, the corresponding window is formally defined as 
follows:
	\begin{table}
		\renewcommand{\arraystretch}{1.5}
		\centering
		\caption{Quantitative semantics values}
		\begin{tabular}{|c|c|c|c|c|c|c|c|c|} \hline
			MTL Formula& \multicolumn{8}{c|}{Time point} \\
			\cline{2-9}&3&4&5&6&7&8&9&10\\ \hline
			$\val{\opP_{[1,4]} p}$&$\frac{1}{4}$&$\frac{1}{2}$&
			$\frac{3}{4}$&1&1&$\frac{3}{4}$&$\frac{1}{2}$&$\frac{1}{4}$\\[5pt]
			
			\hline
		\end{tabular}
		\label{quant_values}
	\end{table}

\begin{equation}
  w^-_{[1, 4]}[i]\,{=}\,\frac{1}{4}\sum\limits_{j=1}^{4} \delta(i-j).
\end{equation}
Now we are able to evaluate the quantitative semantics of signal $x$
with respect to the Once-formula given below:
$$
\begin{array}{lcl}
  \val{\opP_{[1,4]} p}\left[i\right] = 
  \sum\limits_{j=0}^{12}p[j]\cdot w^-_{[1,4]}[i-j].
\end{array}
$$
The quantitative-values of the result are given in a 
Table~\ref{quant_values}, and their graphical representation is shown 
in Figure~\ref{fig:quantitative_exd}.
%The quantitative-values of the result are shown in 
%Figure~\ref{fig:quantitative_exd}.

For the quantitative, discrete-time semantics of an MTL formula
$\varphi$ and signal $x$, it is straight forward to show, that the
semantics is sound, in the sense that, it produces a measure
greater than 0, if and only if, the formula is satisfied by the
discrete-time qualitative semantics, and 0, otherwise.
\begin{figure}
	\centering
	\epsfig{file=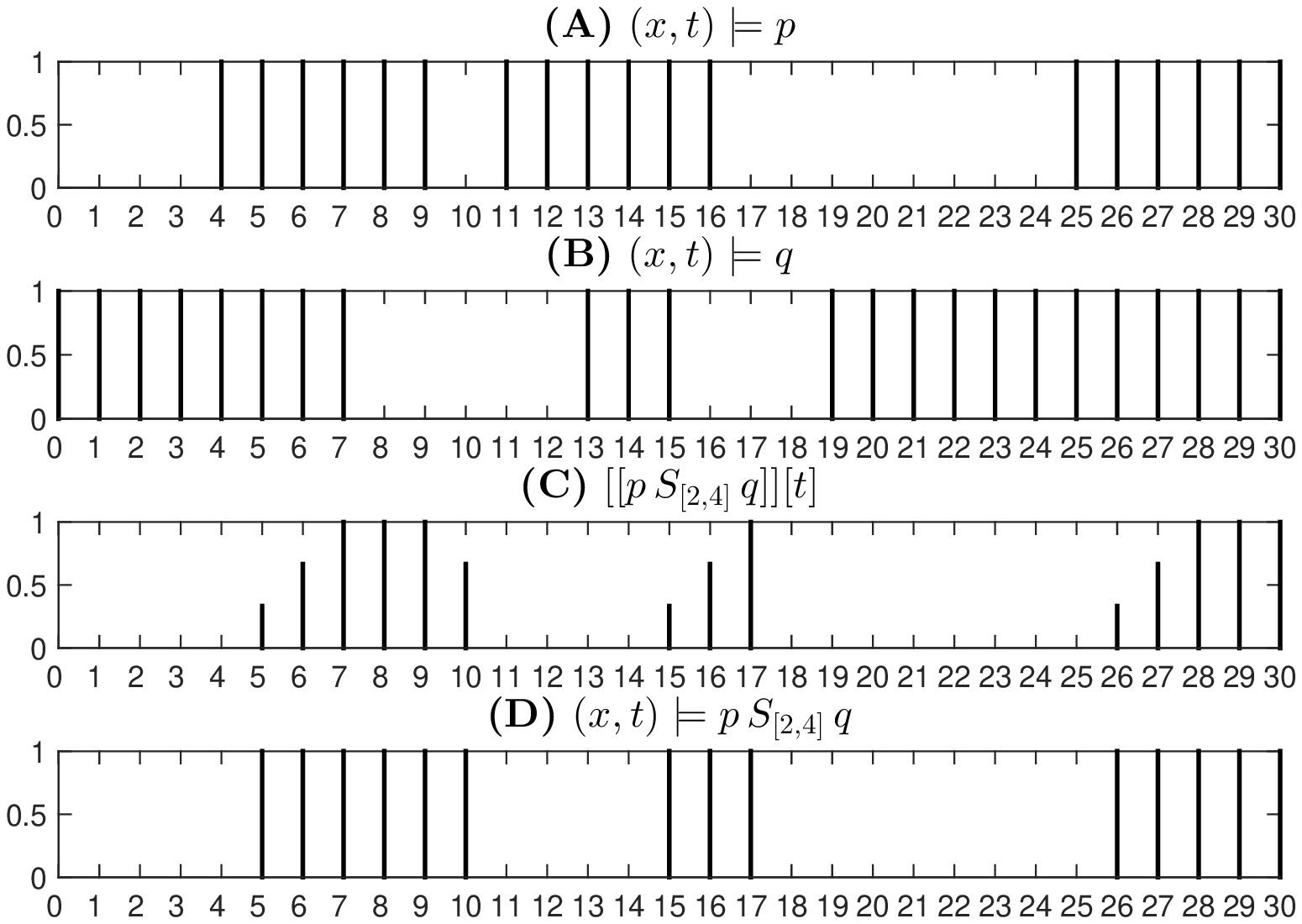, width=\linewidth}
	\caption{Discrete quantitative semantics for $\since_I$.}
	\label{fig:since_d}
\end{figure}

\begin{theorem}[Soundness] 
	For a positive normal form MTL formula $\varphi$ and a 
	discrete-time signal 
	$x$, it holds that $x$ satisfies $\varphi$, if and only if, its 
	quantitative semantics is strictly greater than 0. It does not 
	satisfy 
	$\varphi$, if and only if, its quantitative semantics is 0:
  \begin{align*}
    (x,i) \models\varphi &\iff
    \val{x,\varphi} [i] > 0,\\
    (x,i) \not\models\varphi &\iff
    \val{x,\varphi} [i] = 0.
  \end{align*}
\end{theorem}
\begin{proof} (Sketch)
  First, as shown in Theorem~\ref{thm:qualiDTs}, the qualitative
  discrete-time semantics in terms of LTI-filtering over the
  idempotent dioid $(\max,\min,0,1)$, is the same as the classical 
MTL
  semantics. Second, the interpretation of Boolean MTL formulas $p$,
  $\neg p$, $\varphi\vee\psi$, and $\varphi\wedge\psi$ over $x$
  is the same in both the qualitative and quantitative
  semantics. Third, the interpretation of the temporal MTL formulas
  $\G_I$ and $\opH_I$ over $x$ is also the same. Fourth, the
  interpretation of MTL formulas $\F_I$ and $\opP_I$ ensure that they
  are strictly positive only if they are satisfied by $x$ and 0,
  otherwise. Fifth, the interpretation of $\until_I$ and $\since_I$
  are properly derived from $\G_I$ and $\F_I$, $\opH_I$ and $\opP_I$,
  respectively.
\end{proof}

%\noindent{}
In order to illustrate the relation between the qualitative and the 
quantitative, discrete-time semantics, consider the example
shown in Figure~\ref{fig:thm_d}. In Figure~\ref{fig:thm_d}(A) we show
the signal $x_{p}$.  This signal is defined over the set
$\{0,\ldots,30\}$. In Figure~\ref{fig:thm_d}(B) we illustrate the
quantitative semantics of the formula $\opP_{[3,4]}p$ with respect to
$x$, and in Figure~\ref{fig:thm_d}(C) we show the qualitative
semantics of the same formula with respect to $x$. As
Figures~\ref{fig:thm_d}(B-C) show, the two are in agreement, that is
whenever the quantitative semantics is strictly greater than 0, the
formula is satisfied, and whenever the quantitative semantics is 0,
the formula is not satisfied.

In Figure~\ref{fig:since_d} we illustrate the interplay between the
once and the historically operators as they occur in the definition 
of 
the since operator. Note that the outside sum also plays the role of 
an enclosing temporal operator. In Figure~\ref{fig:since_d}(A) and 
(B)
we show the discrete-time signals $x_{p}$ and $x_q$, respectively. In
Figure~\ref{fig:since_d}(C) we show the quantitative semantics of
$p\since_{[2,4]}q$ with respect to $x_{p}$ and $x_q$. Finally in
Figure~\ref{fig:since_d}(D), we show the qualitative semantics of the
same formula with respect to $x_{p}$ and $x_q$. As one can see from
Figure~\ref{fig:since_d}(C-D), the quantitative semantics is sound.
%%%%%%%%%%%%%%%%%%%%%%%%%%%%%%%%%%%%%%%%%%%%%%%%%%%%%%%%%%%%%%%%%%%%%%
\subsubsection{Continuous-time semantics}
%In the continuous-time semantics, each bounded interval is going to
%have an infinite number of points, except for the singular intervals,
%which contain a single point, only. As a consequence, we have to
%interpret the sum as an integral and the minimum as infimum.  
%
%In the continuous-time semantics, we restrict ourselves to 
%\textit{left-closed right-open} 
%signal segments, therefore we do not consider signals that may have 
%a distinct value in some isolated point.
In the continuous-time semantics, we restrict ourselves to piecewise 
constant \textit{cadlag signals}~\cite{nickovic06}. Such 
signals are 
right-continuous and have left limits everywhere. As a consequence, 
we do not consider signals that may have a distinct value in some 
isolated point.
Since this property has to be preserved by temporal operators, we 
adopt the MTL fragment from~\cite{nickovic06}, where temporal 
modalities are restricted to be closed intervals only.
Assume we are given a signal $x$ with domain $[0,T)$. 
Then the quantitative semantics for MTL is defined as follows:
\begin{align*}
%---------------------------------------------------------------------
% Proposition 
&\val{ x,p } (i) && = &&x_p(i) \\
%---------------------------------------------------------------------
% Negation  
&\val{ x,\neg p } (i) && = &&1 - \val{ x,p }(i) \\
%---------------------------------------------------------------------
% Or 
&\val{ x, \varphi\,{\vee}\,\psi } (i) &&=&& 
\max(\val{ x,\varphi } (i),\ \val{ x,\psi }(i))\\
%---------------------------------------------------------------------
% And
&\val{ x, \varphi\,{\wedge}\,\psi } (i) &&= &&
\min(\val{ x,\varphi } (i),\ \val{ x,\psi }(i))\\
%---------------------------------------------------------------------
% Finally  
&\val{ x, {\F}_{I}\varphi } (i) &&=&& 
\int\limits_{\mathbb{T}} \val{ x,\varphi}(j)\cdot w^+_I(i{-}j) 
\ \mathrm{d}j\\
%---------------------------------------------------------------------
% Once 
&\val{ x, {\opP}_{I}\varphi } (i) &&=&& 
\int\limits_{\mathbb{T}} \val{ x,\varphi}(j)\cdot w^-_I(i{-}j) 
\ \mathrm{d}j \\
%---------------------------------------------------------------------
% Globally 
&\val{ x, {\G}_{I}\varphi } (i) &&=&& 
\inf_{j\,{\in}\,i+I} \val{ x,\varphi}(j)\\
%    \end{align*}
%    \begin{align*}
%---------------------------------------------------------------------
% Historically
&\val{ x,{\opH}_{I}\varphi } (i) &&= &&
\inf_{j\,{\in}\,i-I} \val{ x,\varphi}(j) \\
%---------------------------------------------------------------------
% Until
&\val{ x,\varphi\until_I\psi }(i)&& =&& 
\frac{1}{|I|}\int\limits_I
\val{ x,{\G}_{(0,j)}\varphi }(i)\cdot 
\val{ x,{\F}_{\{j\}}\psi    }(i)
\mathrm{d}j\\
%---------------------------------------------------------------------
% Since
&\val{ x,\varphi\since_I\psi} (i) &&= &&
\frac{1}{|I|}\int\limits_I
\val{ x,{\opH}_{(0,j)}\varphi }(i)\cdot 
\val{ x,{\opP}_{\{j\}}\psi    }(i)
\mathrm{d}j,
\end{align*}
where the window superscripts $+$ and $-$ correspond to the direction
of the window according to (\ref{def:win}).  
%In the
%definition of $\G_I$ and $\opH_I$, $\varphi(j)$ could also have been
%written as the convolution $\varphi\,{*}\,\delta$, for Dirac impulse
%$\delta$. 
As usual, the $\inf$ operation over an empty set is assumed to be 1.

Like in the discrete-time case, it is relatively straight-forward to
show that the quantitative semantics of an MTL formula $\varphi$ with
respect to a continuous-time signal $x$, is sound.
%Due to the integral over a null set is evaluated as zero, such 
%soundness theorem is restricted.  
\begin{figure}
	%\centering
	\epsfig{file=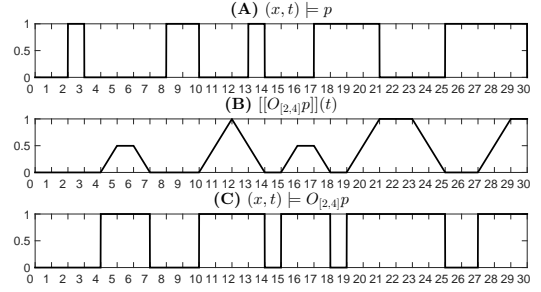, width=\linewidth}
	\caption{Soundness of continuous-time semantics.} 
	\label{fig:theorem_fig}
\end{figure}

\begin{theorem}[Soundness]
%	For a positive normal form MTL formula $\varphi$ and a 
%	continuous-time signal $x$, it holds that $x$ satisfies 
%	$\varphi$, if and only if, its quantitative semantics is strictly 
%	greater than 0. It does not satisfy $\varphi$, if and only if, 
%	its quantitative semantics is 0:
%  \begin{align*}
%    (x,i) \models\varphi & \iff
%    \val{x,\varphi} (i) > 0,\\
%    (x,i) \not\models\varphi & \iff
%    \val{x,\varphi} (i) = 0.
%  \end{align*}
Let $\varphi$ be a positive normal form MTL formula and $x$ be a 
continuous-time signal.
Then the following properties hold:
%Then the following properties relate the quantitative semantics to 
%the qualitative semantics:
	\begin{alignat*}{2}
	&\val{x,\varphi} (i) > 0 \quad &&\Longrightarrow \quad 
	(x,i) \models\varphi,\\
	&(x,i) \not\models\varphi \quad &&\Longrightarrow \quad 
	\val{x,\varphi} (i) = 0.
	\end{alignat*}
\end{theorem}
\begin{proof} (Sketch)
  The proof goes along the same lines as the one for the 
  discrete-time semantics, but with a set of complementary remarks. 
  First, one has to observe that $\inf$ operator properly
  extends $\min$ and the integral properly extends the sum operator.
  Second, the
  Dirac-$\delta$ is the proper substitute of the Kronecker-$\delta$
  within the integral.
  Third, if the MTL formula is valid only in a punctual interval 
  within the given time constraints, it is not allowed to measure the 
  normalized number of times the formula is satisfied within the 
  window. In this case the quantitative 
  measure is equal to zero and it is essential to drop the contrary 
  claim. 
\end{proof}

In order to illustrate the quantitative, continuous-time semantics 
for
various temporal operators, and their relation to the qualitative,
continuous-time semantics, consider the example shown in
Figure~\ref{fig:theorem_fig}. In Figure~\ref{fig:theorem_fig}(A) we
show the signal $x_{p}$.  This signal is defined over the interval
$[0, 30)$.  
%Every sub-interval in the plot is left-closed and right-open. 
In Figure~\ref{fig:theorem_fig}(B) we illustrate the
quantitative semantics of the formula $\opP_{[2,4]}p$ with respect to
$x$, and in Figure~\ref{fig:theorem_fig}(C) we show the qualitative
semantics of the same formula with respect to $x$. As
Figures~\ref{fig:theorem_fig}(B-C) show, the two are in agreement,
that is whenever the quantitative semantics is strictly greater than
0, the formula is satisfied, and whenever formula is not satisfied, 
the quantitative semantics is 0.

\newdef{remark2}{Remark}
\begin{remark2}
Note that for a singular interval $I=\{a\}$, the semantics of
$\F_{I}\varphi$ or $\opP_{I}\varphi$ is not necessarily zero 
(although
a point has zero area), because this interval is represented by the
Dirac-$\delta$ distribution which is an ``infinite'' value at $a$:
\begin{align*}
\val{x, {\opP}_{\{a\}} \varphi}(t)
&=\int\limits_{0}^{T}\val{x, \varphi}(s)\cdot 
w^-_{\{a\}}(t-s)\ \mathrm{d}s = \val{x, \varphi}(t- a), \\
\val{x, {\F}_{\{a\}} \varphi}(t)
&=\int\limits_{0}^{T}\val{x, \varphi}(s)\cdot 
w^+_{\{a\}}(t-s)\ \mathrm{d}s = \val{x, \varphi}(t+ a).
%\val{{\opP}_{\{a\}} \varphi}(t)
%&=\int\limits_{0}^{T}\varphi(s)\cdot 
%w^-_{\{a\}}(t-s)\ \mathrm{d}s = \varphi(t- a), \\
%\val{{\F}_{\{a\}} \varphi}(t)
%&=\int\limits_{0}^{T}\varphi(s)\cdot 
%w^+_{\{a\}}(t-s)\ \mathrm{d}s = \varphi(t+ a).
\end{align*}
%\vspace*{-1ex}
Moreover, the semantics of $\G_I$ and $\opH_I$ is also not 
zero, because of the infimum operator over a single point. As a
consequence, we have that $\G_I\,{\equiv}\,\F_I$ and
$\opH_I\,{\equiv}\,\opP_I$.
\end{remark2}

\begin{figure}
	\centering
	\epsfig{file=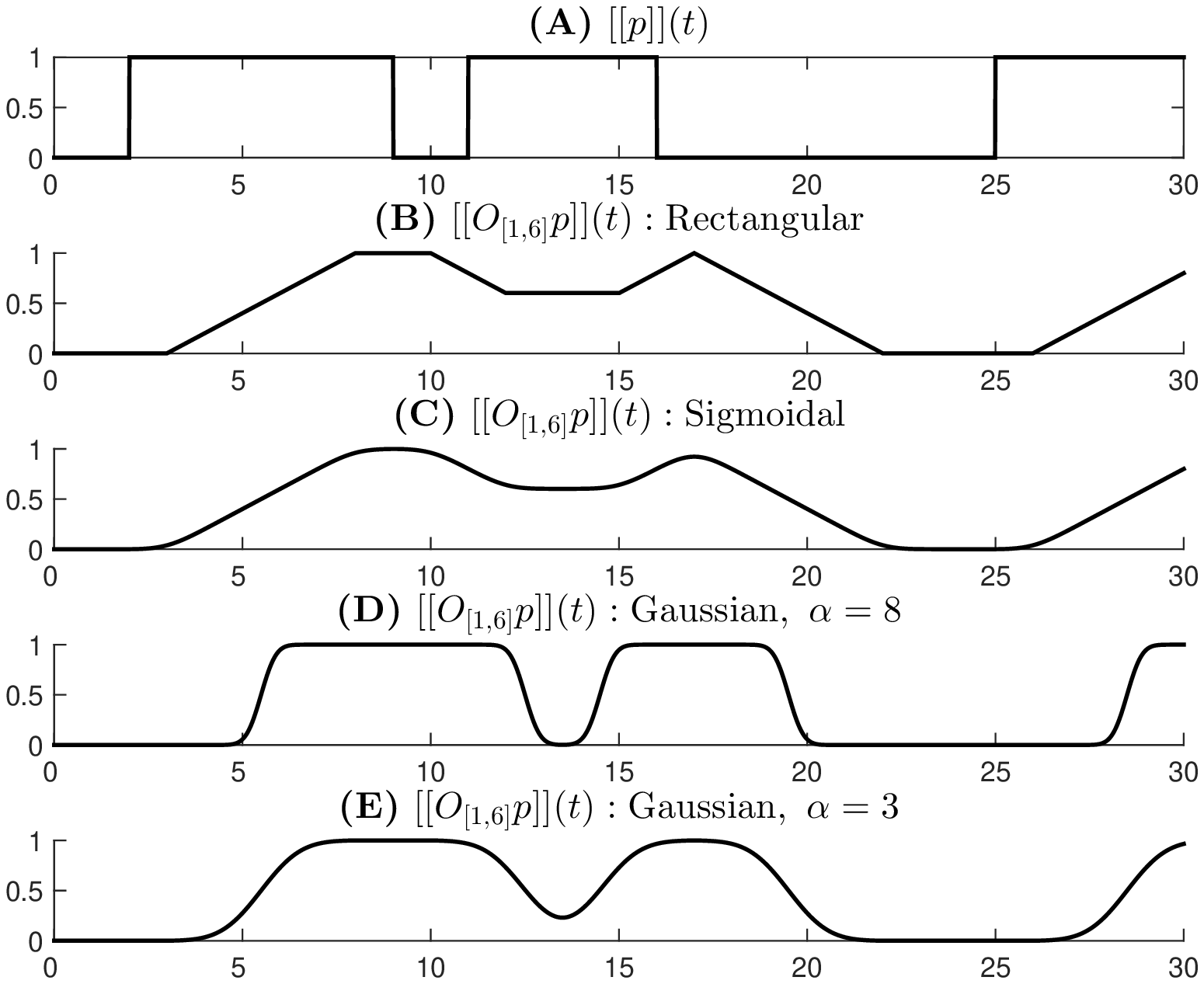, scale=0.42}
	\caption{Smooth evaluation of $\val{\opP_{[1, 6]} p 
		}\left[t\right]$.}
	\label{fig:vs_long16}
\end{figure}
In order to illustrate the effect of applying various smooth kernels
to a signal $x_p$, consider the example shown in
Figure~\ref{fig:vs_long16}. In Figure~\ref{fig:vs_long16}(A) we
show the signal $x_{p}$. In Figure~\ref{fig:vs_long16}(B) we
illustrate the quantitative semantics with respect to a square 
kernel.
In Figure~\ref{fig:vs_long16}(C) we show the same semantics with
respect to a sigmoidal window. This automatically adds tolerance to a
time-jitter at the boundaries of the kernel. Finally, in
Figures~\ref{fig:vs_long16}(D-E) we show the result of applying
Gaussian kernels, with reciprocal of standard deviation 8 and 3,
respectively. Note how they influence the domain of satisfaction.

%
%In Figure~\ref{fig:since_c} we illustrate the interplay between once
%and historical as they occur in the definition of since, for two
%continuous-time signals $x_p$ and $x_q$. In
In Figure~\ref{fig:since_c} we illustrate the since 
operator with respect to continuous-time signals $x_p$ and $x_q$. In
Figures~\ref{fig:since_c}(A-B) we show these two signals. In
Figure~\ref{fig:since_c}(C) we show the quantitative semantics of
$p\since_{[1,3]}q$ with respect to $x_{p}$ and $x_q$. Finally in
Figure~\ref{fig:since_c}(D), we show the qualitative semantics of the
same formula with respect to $x_{p}$ and $x_q$. As one can see from
Figures~\ref{fig:since_c}(C-D), the quantitative semantics is sound.

\begin{figure}
	\centering
	\epsfig{file=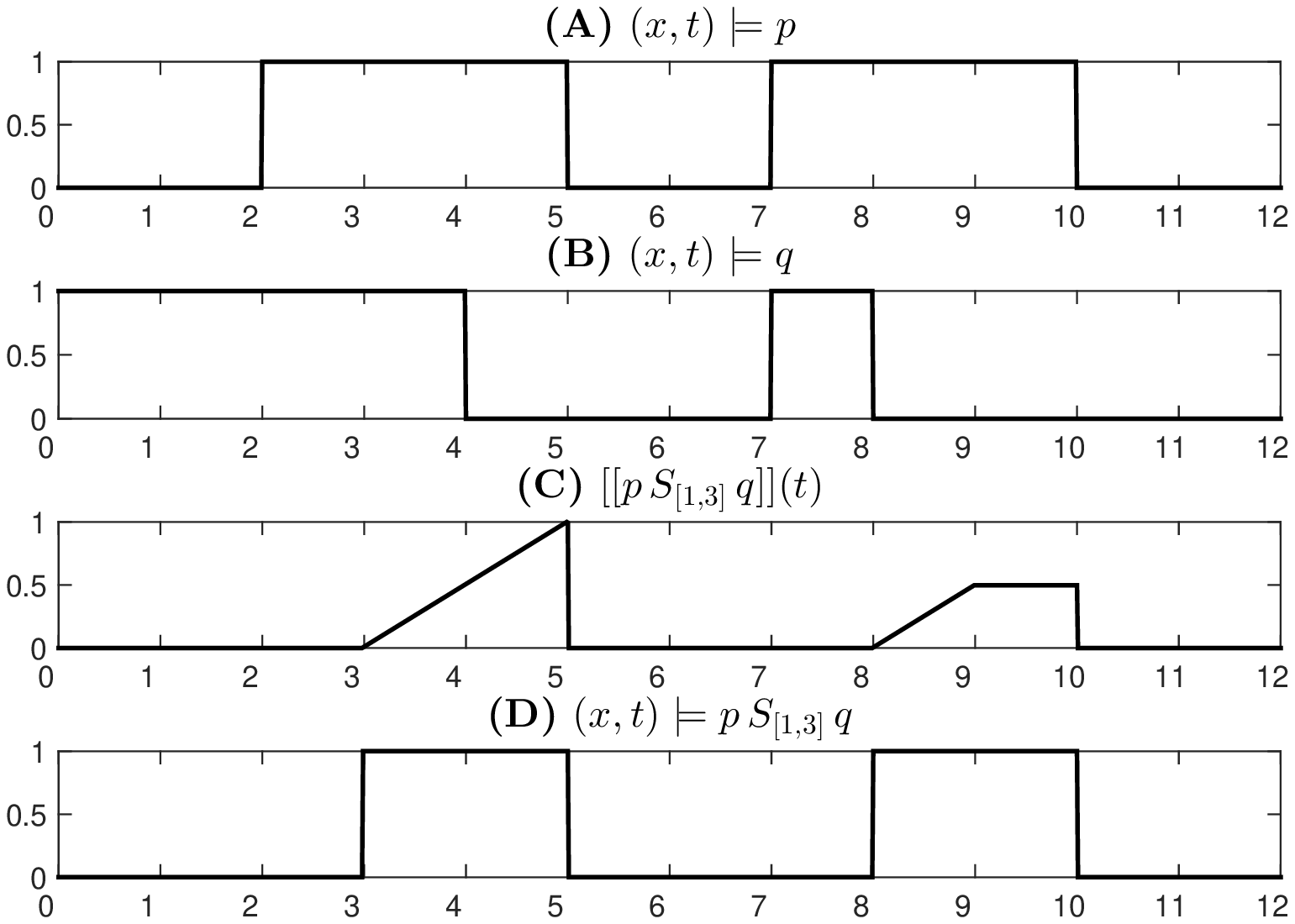, width=\linewidth}
	\caption{Cont.~quantitative semantics for $\since_I$.}
	\label{fig:since_c}
\end{figure}

%\section{Conclusions}
\section{Conclusions}
\label{sec:conclusions}
We have shown that linear, time-invariant (LTI) filtering corresponds
to metric temporal logic (MTL) if addition and multiplication are
interpreted as $\max$ and $\min$, and if true and false are interpreted as
one and zero, respectively.
%To the best of our knowledge, this connection between
%LTI-filtering and MTL has not been established before.

We have also provided a quantitative semantics to temporal MTL formula
with respect to a discrete- or continuous-time signal, measuring the
normalized, maximum number of times, the formula is satisfied within
its associated window. This semantics is sound, in the sense that, if
its measure is strictly greater than zero, then the formula is
satisfied.

In future work we would like to explore alternative quantitative
semantics for MTL, which are more informative, in case the MTL formula
is not satisfied. We would also like to explore the logical meaning of
other types of LTI-filter kernels, such as band-pass and
high-pass. Moreover, we would like to investigate how the
correspondence between MTL and LTI-filters can be exploited in order
to build very efficient MTL monitors, using digital signal processors
(DSPs). These are specialized microprocessors, optimized for filtering
and compression operations in signal processing.

%\cite{EX}

% \input{chapters/thm_proof}

\section{Acknowledgments}

This work was partially supported by the Austrian FFG project HARMONIA
(nr. 845631), the Doctoral Program Logical Methods in Computer Science
funded by the Austrian FWF, the Austrian National Research Network
(nr. S 11405-N23 and S 11412-N23) SHiNE funded by the Austrian Science
Fund (FWF), the EU ICT COST Action IC1402 on Runtime Verification
beyond Monitoring (ARVI), the US National-Science-Foundation Frontiers
project CyberCardia, the MISTRAL project A-1341-RT-GP coordinated by 
the European Defence Agency (EDA) and funded by 
%8 contributing Members (France, Germany, Italy, Poland, Austria, 
%Sweden, Netherlands, Luxembourg) in the framework of 
the Joint Investment Programme on Second Innovative Concepts and 
Emerging Technologies (JIP-ICET 2).
The authors also would like to acknowledge C\'{e}sar S\'{a}nchez from 
IMDEA 
Software Institute for helpful insight and comments that greatly 
improved the paper. 
\bibliographystyle{abbrv}
\bibliography{chapters/sigproc}  % sigproc.bib is the name of the 
%Bibliography in this case
\end{document}